\documentclass{article}
\usepackage[utf8]{inputenc}

\usepackage{amsmath}
\usepackage{amssymb}
\usepackage{amsthm}
\usepackage[toc,page]{appendix}
\usepackage{graphics}

\title{Generalizations of $2$-dimensional diagonal quantum channels with
constant Frobenius norm}
\author{Ivan Sergeev \\ ETH Zürich, Rämistrasse 101, 8092 Zürich
                     \\ e-mail: i.i.sergeyev@gmail.com}

\theoremstyle{definition}
\newtheorem{defn}{Definition}

\theoremstyle{plain}
\newtheorem{lem}{Lemma}
\newtheorem{thm}{Theorem}

\theoremstyle{remark}
\newtheorem{rmk}{Remark}

\begin{document}

\maketitle

\begin{abstract}
We introduce the set of quantum channels with constant Frobenius norm, the set
of diagonal channels and the notion of equivalence of one-parameter families of
channels. First, we show that all diagonal $2$-dimensional channels with
constant Frobenius norm are equivalent. Next, we generalize four one-parameter
families of $2$-dimensional diagonal channels with constant Frobenius norm to an
arbitrary dimension $n$. Finally, we prove that the generalizations are not
equivalent in any dimension $n \ge 3$.
\end{abstract}


\section{Introduction}

Quantum information channels were first introduced by A. S. Holevo in 1972
\cite{Holevo1972}. The definition utilizes the notion of a completely positive
map, introduced by W. F. Stinespring in 1955 \cite{Stinespring1955} and studied
by M.-D. Choi \cite{Choi1975}, K. Kraus \cite{Kraus1983} and others. Although
quantum channels have been studied since the 1970s, the description of the set
of all quantum channels in an arbitrary dimension has not yet been obtained.
The advances toward solving this problem include, for example, a complete
description in dimension $2$ \cite{Ruskai2002} and some generalizations in
higher dimensions \cite{Nathanson2007}.

Many works study specific families of channels and their properties in various
contexts, for example, the depolarizing channel \cite{Amosov2006, Amosov2007pra,
Amosov2009, King2003, Leung2017}, the transpose-depolarizing channel
\cite{Datta2006} (special cases can be found in \cite{Alicki2004, Chen2013,
Horodecki1997, Landau1993, Werner2002}), the sets of ergodic and mixing channels
\cite{Burgarth2013}, Weyl channels \cite{Amosov2009, Amosov2017} etc. In
addition, various functions of the output of quantum channels are often
considered, such as the entropy \cite{Alicki2004, Amosov2006, Amosov2007pra,
Amosov2007, Amosov2009, Amosov2017} and the noncommutative $\ell_{p}$-norms
\cite{Amosov2000, Amosov2001}, especially in relation to the capacity of quantum
channels \cite{Amosov2000, Amosov2007, Holevo1998, King2001, King2003}.

Since it is difficult to describe the set of all quantum channels, we limit the
scope of our work to two narrow sets of channels: the set of quantum channels
with constant output Frobenius norm and the set of diagonal quantum channels.
More specifically, we prove that every $2$-dimensional channel with constant
Frobenius norm is either a completely depolarizing channel or is equivalent to
the depolarizing channel (i.e. differs from it only by a change of basis in the
input and output state spaces). Next, we introduce the notion of diagonal
quantum channels and use it to generalize the four one-parameter families of
channels that emerge from the $2$-dimensional classification. In addition, we
prove that the four generalizations not equivalent in any dimension $n \ge 3$.
In the future, we hope to generalize the results to broader sets of channels.
Although there are papers that deal with functions of the output of quantum
channels (such as those mentioned above), the author is not aware of other works
that consider either the sets of channels with constant output Frobenius norm or
diagonal quantum channels as defined in this paper.

This paper is organized as follows. First, we provide basic definitions in
Section~\ref{sec:basics}. Next, in Section~\ref{sec:dim 2}, we introduce the
notion of equivalence of channels, following \cite{King2001}, and classify all
$2$-dimensional quantum channels with constant Frobenius norm up to equivalence.
Then we move on to the $n$-dimensional case. Section~\ref{sec:diag} is dedicated
to the set of diagonal quantum channels: we define diagonal channels in
Section~\ref{ssec:diag def} and then prove the criterion for them to be channels
with constant Frobenius norm in Section~\ref{ssec:diag f}. In
Section~\ref{ssec:gen def}, we use the notion of diagonal channels to generalize
the four families of channels with constant Frobenius norm that emerge from the
$2$-dimensional case. Then, in Section~\ref{ssec:gen ineq}, the generalizations
are shown to be inequivalent in any dimension $n \ge 3$ (in contrast to the
$2$-dimensional case, where all four channels are equivalent). Finally,
Section~\ref{sec:conclusion} contains discussion of the results and states open
questions for further investigation.

\section{\label{sec:basics}Basic definitions}

In this section, we recall the definitions of quantum channels (following
\cite{Amosov2017, Holevo1980, Holevo2002, Holevo2010}) and of the Frobenius norm
and introduce the set of channels with constant function of the output for pure
input states.

We begin by recalling the notions of adjoint operators and completely positive
maps, which are essential in the definition of quantum channels. Denote the
underlying Hilbert spaces for the input and the output of the channel as
$\mathcal{H}_{1}$ and $\mathcal{H}_{2}$ respectively, the corresponding spaces
of quantum states as
$\mathcal{S}\left(\mathcal{H}_{1}\right)$ and
$\mathcal{S}\left(\mathcal{H}_{2}\right)$,
and the corresponding spaces of observables as
$\mathcal{B}\left(\mathcal{H}_{1}\right)$ and
$\mathcal{B}\left(\mathcal{H}_{2}\right)$.

\begin{defn}
Let $\Phi \colon \mathcal{S}\left(\mathcal{H}_{1}\right) \to
\mathcal{S}\left(\mathcal{H}_{2}\right)$. The adjoint operator
$\Phi^{*} \colon \mathcal{B}\left(\mathcal{H}_{2}\right) \to
\mathcal{B}\left(\mathcal{H}_{1}\right)$
satisfies the following identity:
\begin{equation*}
\mathrm{Tr}\left(S\Phi^{*}\left(A\right)\right) =
\mathrm{Tr}\left(\Phi\left(S\right)A\right) \quad
\forall S\in\mathcal{S}\left(\mathcal{H}_{1}\right), \,
\forall A\in\mathcal{B}\left(\mathcal{H}_{2}\right).
\end{equation*}
In other words, $\Phi^{*}$ is the adjoint of $\Phi$ with respect to the bilinear
form $B\left(A_{1}, \, A_{2}\right) = \mathrm{Tr}\left(A_{1}A_{2}\right)$.
\end{defn}

\begin{defn}
For every $k \in \mathbb{N}$ denote $\mathrm{id}_{k} \colon \mathcal{M}_{k} \to
\mathcal{M}_{k}$ the identity operator on the space of linear operators in
$\mathbb{C}^{k}$. The map $\Phi^{*}$ defined above is called $k$-positive if
$\Phi^{*}\otimes\mathrm{id}_{k}$ is positive, i.e. maps positive elements in
$\mathcal{B}\left(\mathcal{H}_{2}\right) \otimes \mathcal{M}_{k}$
to positive elements in
$\mathcal{B}\left(\mathcal{H}_{1}\right) \otimes \mathcal{M}_{k}$.
The map $\Phi^{*}$ is called completely positive if $\Phi^{*}$ is a $k$-positive
map for every $k \in \mathbb{N}$.
\end{defn}

Using these two notions, we can now formulate the definition of a quantum
channel.

\begin{defn}
A linear map $\Phi \colon \mathcal{S}\left(\mathcal{H}_{1}\right) \to
\mathcal{S}\left(\mathcal{H}_{2}\right)$ is called a quantum channel if its
adjoint $\Phi^{*}$ is a completely positive map.
\end{defn}

The following two theorems give alternative characterizations of completely
positive maps.

\begin{thm}[Choi, see \cite{Choi1975}]
Let $\Phi \colon \mathcal{M}_{n} \to \mathcal{M}_{m}$. Then $\Phi$ is a
completely positive map if and only if the following matrix is positive:
\begin{equation*}
C_{\Phi} = \sum_{i,j=1}^{n}E_{ij} \otimes \Phi\left(E_{ij}\right),
\end{equation*}
where the matrix $E_{ij} \in \mathbb{C}^{n \times n}$ has a $1$ in the $ij$-th
entry and $0$s everywhere else. We refer to $C_{\Phi}$ as the Choi matrix of
$\Phi$.
\end{thm}

\begin{thm}[Kraus, see \cite{Kraus1983}]
Let $\Phi \colon \mathcal{M}_{n} \to \mathcal{M}_{m}$. Then $\Phi$ is a
completely positive map if and only if there exists a set of operators
$\left\{V_{i}\right\} \subseteq \mathbb{C}^{m \times n}$ such that
\begin{equation*}
\Phi\left(S\right) = \sum_{i}V_{i} S V_{i}^{*}.
\end{equation*}
The operators $V_{i}$ are then called the Kraus operators of $\Phi$.
\end{thm}

Next, we define quantum channels with constant function of the output for pure
input states.

\begin{defn}
Let $\Phi \colon \mathcal{S}\left(\mathcal{H}_{1}\right) \to
\mathcal{S}\left(\mathcal{H}_{2}\right)$ be a quantum channel and $f \colon
\mathcal{S}\left(\mathcal{H}_{2}\right) \to \mathbb{R}$ be a function of the
output state. We say that $\Phi$ is a channel with constant function $f$ if for
all pure input states $S \in \mathcal{S}\left(\mathcal{H}_{1}\right)$ the value
$f\left(\Phi\left(S\right)\right)$ is the same.
\end{defn}

Note that lifting the restriction of the input to pure states makes the
resulting set of channels scarcer. For example, by the above definition, the
depolarizing and the transpose-depolarizing channels are channels with constant
Frobenius norm, but they do not have the same output Frobenius norm for mixed
input states as for pure input states. Since we would like to include these
channels into consideration, we restrict the constant output function condition
to pure input states.

In this paper, we assume that both underlying Hilbert spaces are
$n$-dimensional, i.e. $\mathcal{H}_{1} = \mathcal{H}_{2} = \mathbb{C}^{n}$.
In this case, the spaces of quantum states
$\mathcal{S}\left(\mathcal{H}_{1}\right)$ and
$\mathcal{S}\left(\mathcal{H}_{2}\right)$ are isomporphic to the space of
$n$-dimensional positive semidefinite Hermitian operators of trace $1$ and the
spaces of observables
$\mathcal{B}\left(\mathcal{H}_{1}\right)$ and
$\mathcal{B}\left(\mathcal{H}_{2}\right)$ are isomporphic to the space of linear
operators $\mathcal{M}_{n}$. Among the functions of the output state
$f \colon \mathcal{M}_{n} \to \mathbb{R}$ we consider only the Frobenius norm.

\begin{defn}
The Frobenius norm
$\left|\left| \cdot \right|\right|_{F} \colon \mathcal{M}_{n} \to \mathbb{R}$
is a matrix norm defined by the following formula:
\begin{equation*}
\left|\left| S \right|\right|_{F} = \sqrt{\sum_{j=1}^{n}s_{j}^{2}},
\end{equation*}
where $s_{1}, \, \dots, \, s_{n}$ are singular values of $S \in \mathcal{M}_{n}$.
If $S$ is a quantum state, then $s_{j}$ are eigenvalues of $S$.
\end{defn}

Instead of the Frobenius norm, we could use other functions of the output state,
for example, the Schatten norm, the Ky Fan norm (see \cite{Bhatia1996}, section
IV.2), the noncommutative $\ell_{p}$ norm (see \cite{Amosov2001}) or the von
Neumann entropy. However, these and other functions that are invariant with
respect to basis changes in the underlying Hilbert spaces are defined using the
eigenvalues of the output state. Therefore, to compute these functions in
general case, one needs to find the roots of the characteristic polynomial of an
$n \times n$ matrix, i.e. to solve a polynomial equation of degree $n$. On the
other hand, one can compute the Frobenius norm without finding the eigenvalues
by using the following identities:
\begin{equation*}
\left|\left| S \right|\right|_{F} =
\sqrt{\sum_{i,j=1}^{n}\left| s_{ij} \right|^{2}} =
\sqrt{\mathrm{Tr}\left(S^{*}S\right)},
\end{equation*}
where $s_{ij}$ denote the entries of $S$ and $S^{*}$ is the adjoint of $S$. For
this reason, we choose to only consider the Frobenius norm in this work.

\section{\label{sec:dim 2}Dimension $2$}

In this section, we introduce the notion of equivalence of quantum channels and
classify all $2$-dimensional quantum channels with constant Frobenius norm up to
equivalence.

King and Ruskai \cite{King2001} showed that every $2$-dimensional channel admits
the following representation:
\begin{equation} \label{eq:ruskai repr}
\Phi\left(S\right) = U_{2}\Lambda\left(U_{1} S U_{1}^{*}\right)U_{2}^{*},
\end{equation}
where $U_{1}$ and $U_{2}$ are unitary operators and $\Lambda$ has the following
representation in the Pauli basis
$\left\{I, \, \sigma_{x}, \, \sigma_{y}, \, \sigma_{z}\right\}$:
\begin{equation*}
\Lambda = \left(\begin{array}{cccc}
1     & 0           & 0           & 0 \\
t_{x} & \lambda_{x} & 0           & 0 \\
t_{y} & 0           & \lambda_{y} & 0 \\
t_{z} & 0           & 0           & \lambda_{z}
\end{array}\right),
\end{equation*}
where $t_{\alpha}, \, \lambda_{\alpha} \in \mathbb{R}$ for every
$\alpha \in \left\{x, \, y, \, z\right\}$.

As a direct generalization of this representation, we introduce the following
definition.

\begin{defn}
Let $\Phi_{1} \colon \mathcal{S}\left(\mathcal{H}_{1}\right) \to
\mathcal{S}\left(\mathcal{H}_{2}\right)$ and $\Phi_{2} \colon
\mathcal{S}\left(\mathcal{H}_{1}\right) \to
\mathcal{S}\left(\mathcal{H}_{2}\right)$ be two quantum channels. We call these
channels equivalent and write $\Phi_{1} \sim \Phi_{2}$ if there exist unitary
operators $U_{1} \colon \mathcal{S}\left(\mathcal{H}_{1}\right) \to
\mathcal{S}\left(\mathcal{H}_{1}\right)$ and $U_{2} \colon
\mathcal{S}\left(\mathcal{H}_{2}\right) \to
\mathcal{S}\left(\mathcal{H}_{2}\right)$ such that for every input state $S \in
\mathcal{S}\left(\mathcal{H}_{1}\right)$ the following identity holds:
\begin{equation*}
\Phi_{1}\left(S\right) = U_{2} \Phi_{2}\left(U_{1} S U_{1}^{*}\right) U_{2}^{*}.
\end{equation*}
In other words, equivalent channels can be said to be the same channel written
in two forms that differ by basis changes in the underlying Hilbert spaces.

Additionally, we say that two one-parameter families of channels
$\Phi_{1}\left(p,\,S\right)$ with parameter $p\in\left[p_{1},\,p_{2}\right]$ and
$\Phi_{2}\left(\tilde{p},\,S\right)$ with parameter
$\tilde{p}\in\left[\tilde{p}_{1},\,\tilde{p}_{2}\right]$ are equivalent if for
every $p\in\left[p_{1},\,p_{2}\right]$ there exists
$\tilde{p}\in\left[\tilde{p}_{1},\,\tilde{p}_{2}\right]$ such that the channels
$\Phi_{1}\left(p,\,S\right)$ and $\Phi_{2}\left(\tilde{p},\,S\right)$ (where $p$
and $\tilde{p}$ are fixed) are equivalent.
\end{defn}

The following theorem classifies $2$-dimensional quantum channels with constant
Frobenius norm up to equivalence. Although it can be derived from the
classification of all qubit channels \cite{Ruskai2002}, we provide its full
proof that only employs the notion of equivalence for the sake of completeness.

\begin{thm} \label{thm:2d classification}
Every $2$-dimensional quantum channel with constant Frobenius norm is either a
completely depolarizing channel or is equivalent to the depolarizing channel
\begin{equation*}
\Phi_{d}\left(p, \, S\right) = pS + \frac{1-p}{2}\mathrm{Tr} \, S.
\end{equation*}
\end{thm}

\begin{proof}
Let $S$ be an arbitrary input state. Then the state $S' = U_{1} S U_{1}^{*}$ and
its image $\Lambda\left(S'\right)$ can be written in the following form:
\begin{equation*}
S' = \frac{1}{2}\left(I
   + a_{x}\sigma_{x} + a_{y}\sigma_{y} + a_{z}\sigma_{z}\right)
   = \frac{1}{2}\left(\begin{array}{cc}
     1 + a_{z}       & a_{x} - i a_{y} \\
     a_{x} + i a_{y} & 1 - a_{z} \end{array}\right),
\end{equation*}
\begin{equation*}
\Lambda\left(S'\right) = \frac{1}{2}\Lambda\left(I
   + a_{x}\sigma_{x} + a_{y}\sigma_{y} + a_{z}\sigma_{z}\right)
   = \frac{1}{2}\left(I +
     \sum_{i=x,y,z}\left(t_{i} + a_{i}\lambda_{i}\right)\sigma_{i}\right).
\end{equation*}
Here $a_{x}, \, a_{y}, \, a_{z} \in \mathbb{R}$ are the Stokes parameters. Note
that the condition $\mathrm{Tr} \, S' = 1$ is already satisfied and the
condition $\det S' \ge 0$ results in the following inequality:
\begin{equation*}
a_{x}^{2} + a_{y}^{2} + a_{z}^{2} \le 1.
\end{equation*}
Since the Frobenius norm is unitarily invariant,
\begin{align}
\left|\left| \Phi\left(S\right) \right|\right|_{F}^{2}
 & = \left|\left| U_{2}\Lambda\left(U_{1} S U_{1}^{*}\right)U_{2}^{*}
     \right|\right|_{F}^{2}
   = \left|\left| \Lambda\left(U_{1} S U_{1}^{*}\right) \right|\right|_{F}^{2}
   = \nonumber \\
 & = \frac{1}{2}\left(1 + \left(t_{x} + a_{x}\lambda_{x}\right)^{2} +
     \left(t_{y} + a_{y}\lambda_{y}\right)^{2} +
     \left(t_{z} + a_{z}\lambda_{z}\right)^{2}\right).
     \label{eq:2d ch f norm}
\end{align}
Therefore, $\Phi$ is a channel with constant Frobenius norm if and only if
\eqref{eq:2d ch f norm} is a constant polynomial of $a_{x}$, $a_{y}$ and
$a_{z}$, i.e. all combinations of $a_{x}$, $a_{y}$ and $a_{z}$ in
\eqref{eq:2d ch f norm} have zero coefficients. For a pure input state $S$, we
can see that $S'$ is also a pure state and thus
\begin{equation*}
a_{x}^{2} + a_{y}^{2} + a_{z}^{2} = 1.
\end{equation*}
It follows that \eqref{eq:2d ch f norm} is a constant polynomial in only two
cases.
\begin{enumerate}
\item $\lambda_{x} = \lambda_{y} = \lambda_{z} = 0$ and $t_{x}$, $t_{y}$, and
$t_{z}$ satisfy the following inequality:
\begin{equation*}
t_{x}^{2} + t_{y}^{2} + t_{z}^{2} \le 1.
\end{equation*}
In this case, $\Phi$ is a completely depolarizing channel.
\item $t_{x} = t_{y} = t_{z} = 0$ and $\left| \lambda_{x} \right| =
\left| \lambda_{y} \right| = \left| \lambda_{z} \right| = p$. In this case,
$\Phi$ is a diagonal channel. Without loss of generality, assume
$\lambda_{z} = p$. Then there are four variants:
\begin{enumerate}
\item $\lambda_{x} = \lambda_{y} = p$,
\item $\lambda_{x} = p$ and $\lambda_{y} = -p$,
\item $\lambda_{x} = \lambda_{y} = -p$,
\item $\lambda_{x} = -p$ and $\lambda_{y} = p$.
\end{enumerate}
Denote the corresponding channels $\Phi_{1}$--$\Phi_{4}$. Note that $\Phi_{1}$
is in fact the depolarizing channel:
\begin{equation*}
\Phi_{1}\left(p, \, S\right) =
\frac{1}{2}\left(I + pa_{x}\sigma_{x} + pa_{y}\sigma_{y} +
                 pa_{z}\sigma_{z}\right) =
\Phi_{d}\left(S\right).
\end{equation*}
The following identities show that the other three channels are equivalent to
the depolarizing channel:
\begin{align*}
\sigma_{y} \Phi_{2}\left(-p, \, S\right) \sigma_{y}
 & = \frac{1}{2}\left(\sigma_{y}\sigma_{y}
   - pa_{x}\sigma_{y}\sigma_{x}\sigma_{y}
   + pa_{y}\sigma_{y}\sigma_{y}\sigma_{y}
   - pa_{z}\sigma_{y}\sigma_{z}\sigma_{y}\right) = \\
 & = \frac{1}{2}\left(I
   + pa_{x}\sigma_{x} + pa_{y}\sigma_{y} + pa_{z}\sigma_{z}\right)
   = \Phi_{d}\left(p, \, S\right),
\end{align*}
\begin{align*}
\sigma_{z} \Phi_{3}\left(p, \, S\right) \sigma_{z}
 & = \frac{1}{2}\left(\sigma_{z}\sigma_{z}
   - pa_{x}\sigma_{z}\sigma_{x}\sigma_{z}
   - pa_{y}\sigma_{z}\sigma_{y}\sigma_{z}
   + pa_{z}\sigma_{z}\sigma_{z}\sigma_{z}\right) = \\
 & = \frac{1}{2}\left(I
   + pa_{x}\sigma_{x} + pa_{y}\sigma_{y} + pa_{z}\sigma_{z}\right)
   = \Phi_{d}\left(p, \, S\right),
\end{align*}
\begin{align*}
\sigma_{x} \Phi_{4}\left(-p, \, S\right) \sigma_{x}
 & = \frac{1}{2}\left(\sigma_{x}\sigma_{x}
   + pa_{x}\sigma_{x}\sigma_{x}\sigma_{x}
   - pa_{y}\sigma_{x}\sigma_{y}\sigma_{x}
   - pa_{z}\sigma_{x}\sigma_{z}\sigma_{x}\right) = \\
 & = \frac{1}{2}\left(I
   + pa_{x}\sigma_{x} + pa_{y}\sigma_{y} + pa_{z}\sigma_{z}\right)
   = \Phi_{d}\left(p, \, S\right).
\end{align*}
\end{enumerate}
\end{proof}

\begin{rmk}
Note that the theorem also holds for channels with constant von Neumann entropy,
Schatten $p$-norm (for $1 < p \le \infty$) and Ky Fan $1$-norm, since the same
proof works if the Frobenius norm is replaced with any injective function.

Since the Schatten $1$-norm and the Ky Fan $2$-norm in dimension $2$ give the
trace of the matrix, channels with constant Schatten $1$-norm and Ky Fan
$2$-norm constitute the whole set of $2$-dimensional channels, which was fully
described in \cite{Ruskai2002}.
\end{rmk}

\section{\label{sec:diag}Diagonal channels}

In this section, we introduce the set of diagonal quantum channels and prove the
criterion for diagonal channels to be channels with constant Frobenius norm,
which is used in Section~\ref{ssec:gen def} to define generalizations of the
four one-parameter families of channels that appear in the proof of
Theorem~\ref{thm:2d classification}.

\subsection{\label{ssec:diag def}Definition}

In essence, diagonal channels are defined as channels that have diagonal form in
a specific basis in $\mathcal{M}_{n}$. The basis we use consists of Hermitian
$n \times n$ matrices and is orthonormal with respect to the bilinear form
$B\left(A_{1}, \, A_{2}\right) = \mathrm{Tr}\left(A_{1}^{*}A_{2}\right)$. There
are two reasons why we choose a basis with these properties. Firstly, the space
of quantum states $\mathcal{S}\left(\mathbb{C}^{n}\right)$ is isomporhic space
of $n$-dimensional positive semidefinite Hermitian operators of trace $1$, hence
the Hermitian requirement. Secondly, we would like all decompositions of quantum
states over our basis to have real coefficients, hence the orthonormality
condition.

Before defining an orthonormal basis of Hermitian matrices in $\mathcal{M}_{n}$,
let us consider the $2$-dimensional case. Recall that the Pauli matrices
$\sigma_x, \, \sigma_y, \, \sigma_z$ together with the identity matrix $I$ form
an orthogonal basis in the space of $2 \times 2$ Hermitian matrices. Therefore,
we introduce an intuitive $n$-dimensional generalization of the Pauli matrices
as a starting point for defining our basis in $\mathcal{M}_{n}$.

\begin{defn}
Denote $N = \frac{n\left(n-1\right)}{2}$. We define generalized Pauli matrices
in dimension $n$ as follows:
\begin{equation*}\scalebox{.7}{$
\sigma_{0,1} = I = \left(\begin{array}{cccc}
1     & 0     & \dots & 0 \\
0     & 1     & \dots & 0 \\
\dots & \dots & \dots & \dots \\
0     & 0     & \dots & 1
\end{array}\right);
$}\end{equation*}
\begin{equation*}\scalebox{.7}{$
\sigma_{x,1} = \left(\begin{array}{cccc}
0     & 1     & \dots & 0 \\
1     & 0     & \dots & 0 \\
\dots & \dots & \dots & \dots \\
0     & 0     & \dots & 0
\end{array}\right), \,
\sigma_{x,2} = \left(\begin{array}{ccccc}
0     & 0     & 1     & \dots & 0 \\
0     & 0     & 0     & \dots & 0 \\
1     & 0     & 0     & \dots & 0 \\
\dots & \dots & \dots & \dots & \dots \\
0     & 0     & 0     & \dots & 0
\end{array}\right), \, \dots, \,
\sigma_{x,N} = \left(\begin{array}{cccc}
0     & \dots & 0     & 0 \\
\dots & \dots & \dots & \dots \\
0     & \dots & 0     & 1 \\
0     & \dots & 1     & 0
\end{array}\right);
$}\end{equation*}
\begin{equation*}\scalebox{.7}{$
\sigma_{y,1} = \left(\begin{array}{cccc}
0     & -i    & \dots & 0 \\
i     & 0     & \dots & 0 \\
\dots & \dots & \dots & \dots \\
0     & 0     & \dots & 0
\end{array}\right), \,
\sigma_{y,2} = \left(\begin{array}{ccccc}
0     & 0     & -i    & \dots & 0 \\
0     & 0     & 0     & \dots & 0 \\
i     & 0     & 0     & \dots & 0 \\
\dots & \dots & \dots & \dots & \dots \\
0     & 0     & 0     & \dots & 0
\end{array}\right), \, \dots, \,
\sigma_{y,N} = \left(\begin{array}{cccc}
0     & \dots & 0     & 0 \\
\dots & \dots & \dots & \dots \\
0     & \dots & 0     & -i \\
0     & \dots & i     & 0
\end{array}\right);
$}\end{equation*}
\begin{equation*}\scalebox{.7}{$
\sigma_{z,1} = \left(\begin{array}{cccc}
1     & 0     & \dots & 0 \\
0     & -1    & \dots & 0 \\
\dots & \dots & \dots & \dots \\
0     & 0     & \dots & 0
\end{array}\right), \,
\sigma_{z,2} = \left(\begin{array}{ccccc}
1     & 0     & 0     & \dots & 0 \\
0     & 0     & 0     & \dots & 0 \\
0     & 0     & -1    & \dots & 0 \\
\dots & \dots & \dots & \dots & \dots \\
0     & 0     & 0     & \dots & 0
\end{array}\right), \, \dots, \,
\sigma_{z,N} = \left(\begin{array}{cccc}
0     & \dots & 0     & 0 \\
\dots & \dots & \dots & \dots \\
0     & \dots & 1     & 0 \\
0     & \dots & 0     & -1
\end{array}\right).
$}\end{equation*}
\end{defn}

Note that by construction, the generalized Pauli matrices are Hermitian and
orthogonal with respect to the bilinear form $B\left(\cdot, \, \cdot\right)$
defined above. However, for $n \ge 3$ the matrices
$\sigma_{z,1}, \, \dots, \, \sigma_{z,N}$ are not linearly independent.
Therefore, in the definition of our basis, we replace them with other matrices.

\begin{defn}
Denote
\begin{equation*}
\mathcal{E} = \left\{ \frac{I}{\sqrt{n}},
\frac{\sigma_{x,1}}{\sqrt{2}}, \, \dots, \, \frac{\sigma_{x,N}}{\sqrt{2}}, \,
\frac{\sigma_{y,1}}{\sqrt{2}}, \, \dots, \, \frac{\sigma_{y,N}}{\sqrt{2}}, \,
\frac{M_{z,1}}{\sqrt{2}}, \, \frac{M_{z,2}}{\sqrt{6}}, \, \dots, \,
\frac{M_{z,n-1}}{\sqrt{\left(n-1\right)n}} \right\},
\end{equation*}
where $N = \frac{n\left(n-1\right)}{2}$ as before,
$\left\{ \sigma_{x,i}\right \}_{i=1}^{N}$ and
$\left\{ \sigma_{y,i}\right \}_{i=1}^{N}$ are generalized Pauli matrices, and
\begin{equation*}
M_{z,1} = \left(\begin{array}{cccc}
1     & 0     & \dots & 0 \\
0     & -1    & \dots & 0 \\
\dots & \dots & \dots & \dots \\
0     & 0     & \dots & 0
\end{array}\right), \,
M_{z,2}=\left(\begin{array}{ccccc}
1     & 0     & 0     & \dots & 0 \\
0     & 1     & 0     & \dots & 0 \\
0     & 0     & -2    & \dots & 0 \\
\dots & \dots & \dots & \dots & \dots \\
0     & 0     & 0     & \dots & 0
\end{array}\right), \, \dots,
\end{equation*}
\begin{equation*}
M_{z,n-1} = \left(\begin{array}{cccc}
1     & \dots & 0     & 0 \\
\dots & \dots & \dots & \dots \\
0     & \dots & 1     & 0 \\
0     & \dots & 0     & -\left(n-1\right)
\end{array}\right).
\end{equation*}
\end{defn}

Note that, by construction, $\mathcal{E}$ is a basis in $\mathcal{M}_{n}$ over
$\mathbb{C}$, consists of Hermitian matrices, and is orthonormal with respect
to the bilinear form $B\left(\cdot, \, \cdot\right)$. However, there are also
some disadvantages to using the basis $\mathcal{E}$. For example, a permutation
of basis vectors of the underlying Hilbert space $\mathbb{C}^{n}$ results in a
permutation of the generalized Pauli matrices, but the matrices
$e_{z,1}, \, \dots, \, e_{z,n-1}$ are mapped to nontrivial linear combinations
of the elements of the basis $\mathcal{E}$.

Now, using the basis $\mathcal{E}$, we define the set of diagonal channels.

\begin{defn}
We say that $\Phi \colon \mathcal{S}\left(\mathbb{C}^{n}\right) \to
\mathcal{S}\left(\mathbb{C}^{n}\right)$ is a diagonal channel if $\Phi$ has
diagonal form in the basis $\mathcal{E} = \left\{e_{i}\right\}_{i=0}^{n^2-1}$,
i.e.
\begin{equation} \label{eq:diag ch def}
\Phi = \mathrm{diag}\left(1, \, t_{1}, \, \dots, \, t_{n^{2}-1}\right).
\end{equation}

Note that since $\mathcal{E}$ is an orthonormal basis of Hermitian matrices, the
adjoint operator $\Phi^{*}$ of a diagonal channel $\Phi$ has the same form in
$\mathcal{E}$:
\begin{equation*}
\Phi^{*} = \mathrm{diag}\left(1, \, t_{1}, \, \dots, \, t_{n^{2}-1}\right).
\end{equation*}
In this sense, diagonal channels can be said to be self-adjoint.

When appropriate, we also use the following notation for the basis elements and
diagonal channels:
\begin{align*}
\mathcal{E} & = \left\{e_{0,1}, \, e_{x,1}, \, \dots, \, e_{x,N}, \,
    e_{y,1}, \, \dots, \, e_{y,N}, \, e_{z,1}, \, \dots, \, e_{z,n-1}\right\},\\
\Phi & = \mathrm{diag}\left(1, \, t_{x,1}, \, \dots, \, t_{x,N}, \,
    t_{y,1}, \, \dots, \, t_{y,N}, \, t_{z,1}, \, \dots, \, t_{z,n-1}\right).
\end{align*}
\end{defn}

\subsection{\label{ssec:diag f}Diagonal channels with constant Frobenius norm}

The following theorem gives a criterion for diagonal channels to be channels
with constant Frobenius norm.

\begin{thm} \label{thm:diag ch cfn crit}
Let $\Phi \colon \mathcal{S}\left(\mathbb{C}^{n}\right) \to
\mathcal{S}\left(\mathbb{C}^{n}\right)$ be a diagonal channel.
Then $\Phi$ is a channel with constant Frobenius norm if and only if
$\left| t_{i} \right| = \left| t_{j} \right|$
for all $i, \, j \in \overline{1, \, n^{2}-1}$.
\end{thm}

\begin{proof}
First, let us calculate the Frobenius norm of an arbitrary state
$S \in \mathcal{S}\left(\mathbb{C}^{n}\right)$ and its image
$\Phi\left(S\right)$. Note that $S$ can be represented as
\begin{equation*}
S = \frac{1}{\sqrt{n}}e_{0,1} + \sum_{i=1}^{n^{2}-1}a_{i}e_{i},
\end{equation*}
where $e_{i} \in \mathcal{E}$ are basis matrices. If $\Phi$ is a diagonal
channel of the form \eqref{eq:diag ch def}, then
\begin{equation*}
\left|\left| S \right|\right|_{F}^{2}
 = \frac{1}{n} + \sum_{i=1}^{n^{2}-1}a_{i}^{2}, \quad
\left|\left| \Phi\left(S\right) \right|\right|_{F}^{2}
 = \frac{1}{n} + \sum_{i=1}^{n^{2}-1}a_{i}^{2}t_{i}^{2}.
\end{equation*}
Additionally, if $S$ is a pure state, then
\begin{equation*}
\left|\left| S \right|\right|_{F}^{2}
 = \frac{1}{n}+\sum_{i=1}^{n^{2}-1}a_{i}^{2} = 1.
\end{equation*}

Now, suppose $\Phi$ is a diagonal channel with
$\left| t_{i} \right| = \left| t_{j} \right|$ for all
$i, \, j \in \overline{1, \, n^{2}-1}$.
Then for any pure state $S$ we have
\begin{equation*}
\left|\left| \Phi\left(S\right) \right|\right|_{F}^{2}
 = \frac{1}{n} + t_{1}^{2}\left(1-\frac{1}{n}\right) = \mathrm{const}.
\end{equation*}
Thus $\Phi$ is a channel with constant Frobenius norm.

Conversely, suppose $\Phi$ is a diagonal channel with constant Frobenius norm.
We prove that $\left| t_{i} \right| = \left|t_{j}\right|$ for all
$i, \, j \in \overline{1, \, n^{2}-1}$ in three steps.
\begin{enumerate}
\item First, we prove that
$\left| t_{z,i} \right| = \left| t_{z,j} \right| =: \left| t_{z} \right|$
for all $i, \, j \in \overline{1, \, n-1}$. To show this, we introduce the
following set of unit vectors in $\mathbb{C}^{n}$:
\begin{equation*}
\left| \psi_{0} \right\rangle
 = \left(\begin{array}{c} 1 \\ 0 \\ \dots \\ 0 \end{array}\right), \quad
\left| \psi_{1} \right\rangle
 = \left(\begin{array}{c} 0 \\ 1 \\ 0 \\ \dots \\ 0 \end{array}\right), \quad
\dots, \quad \left| \psi_{n-1} \right\rangle
 = \left(\begin{array}{c} 0 \\ \dots \\ 0 \\ 1 \end{array}\right).
\end{equation*}
Using these vectors, we define a set of pure states
$S_{k} = \left| \psi_{k} \right> \left< \psi_{k} \right|$
with the following decompositions in $\mathcal{E}$:
\begin{equation*}
S_{k} = \frac{1}{\sqrt{n}}e_{0,1} +
\sum_{i=1}^{n-1}B\left(e_{z,i}, \, S_{k}\right)e_{z,i}.
\end{equation*}
For every $k \in \overline{1, \, n-1}$ we have
\begin{equation*}
B\left(e_{z,i}, \, S_{k}\right)
 = \frac{B\left(M_{z,i}, \, S_{k}\right)}{\sqrt{i\left(i+1\right)}}
 = \frac{1}{\sqrt{i\left(i+1\right)}} \cdot \begin{cases}
     0,  & i<k\\
     -i, & i=k\\
     1,  & i>k
   \end{cases}.
\end{equation*}
Since $\Phi$ is a channel with constant Frobenius norm, for every
$k \in \overline{1, \, n-1}$ we have
\begin{equation*}
0 = \left|\left| \Phi\left(S_{k}  \right) \right|\right|_{F}^{2}
  - \left|\left| \Phi\left(S_{k+1}\right) \right|\right|_{F}^{2}
  = t_{z,k}^{2} \cdot \frac{k}{k+1} - t_{z,k+1}^{2} \cdot \frac{k}{k+1}.
\end{equation*}
Therefore, $\left| t_{z,k} \right| = \left| t_{z,k+1} \right|$ for every
$k \in \overline{1, \, n-1}$. Thus
$\left| t_{z,i} \right| = \left|t_{z,j} \right| =: \left| t_{z} \right|$
for all $i, \, j \in \overline{1, \, n-1}$.

\item Second, we prove that
$\left| t_{x,i} \right| = \left| t_{y,i} \right|$ for all
$i \in \overline{1, \, N}$. To show this, we introduce another two sets of unit
vectors in $\mathbb{C}^{n}$:
\begin{equation*}\scalebox{.8}{$
\left| \xi_{1} \right\rangle = \frac{1}{\sqrt{2}}
\left(\begin{array}{c} 1 \\ 1 \\ 0 \\ \dots \\ 0 \end{array}\right), \,
\left| \eta_{1} \right\rangle =\frac{1}{\sqrt{2}}
\left(\begin{array}{c} i \\ 1 \\ 0 \\ \dots \\ 0 \end{array}\right), \,\dots,\,
\left| \xi_{N} \right\rangle = \frac{1}{\sqrt{2}}
\left(\begin{array}{c} 0 \\ \dots \\ 0 \\ 1 \\ 1 \end{array}\right), \,
\left| \eta_{N} \right\rangle = \frac{1}{\sqrt{2}}
\left(\begin{array}{c} 0 \\ \dots \\ 0 \\ i \\ 1 \end{array}\right).
$}\end{equation*}
Using these vectors, we define two sets of pure states
$\tilde{S}_{k} = \left| \xi_{k} \right\rangle \left\langle \xi_{k} \right|$
and $S_{k}' = \left| \eta_{k} \right\rangle \left\langle \eta_{k} \right|$
with the following decompositions in $\mathcal{E}$:
\begin{align*}
\tilde{S}_{k}
 & = \frac{1}{\sqrt{n}}e_{0,1}
   + \sum_{i=1}^{n-1}B\left(e_{z,i}, \, \tilde{S}_{k}\right)e_{z,i}
   + \frac{1}{\sqrt{2}}e_{x,k}, \\
S_{k}'
 & = \frac{1}{\sqrt{n}}e_{0,1}
   + \sum_{i=1}^{n-1}B\left(e_{z,i}, \, S_{k}'\right)e_{z,i}
   + \frac{1}{\sqrt{2}}e_{y,k}.
\end{align*}
For every $k \in \overline{1, \, N}$ and every $i \in \overline{1, \, n-1}$, we
have
\begin{equation*}
B\left(e_{z,i}, \, \tilde{S}_{k}\right) = B\left(e_{z,i}, \, S_{k}'\right).
\end{equation*}
Since $\Phi$ is a channel with constant Frobenius norm, for every
$k \in \overline{1, \, N}$ we have
\begin{equation*}
0 = \left|\left| \Phi\left(\tilde{S}_{k}\right) \right|\right|_{F}^{2}
  - \left|\left| \Phi\left(S_{k}'\right) \right|\right|_{F}^{2}
  = t_{x,k}^{2} \cdot \left(\frac{1}{\sqrt{2}}\right)^{2}
  - t_{y,k}^{2} \cdot \left(\frac{1}{\sqrt{2}}\right)^{2}.
\end{equation*}
Therefore, $\left| t_{x,k} \right| = \left| t_{y,k} \right|$ for all
$k \in \overline{1, \, N}$.

\item Finally, we prove that
$\left| t_{x,i} \right| = \left| t_{z} \right|$ for all
$i \in \overline{1, \, N}$. To show this, we use the pure states $S_{k}$ and
$\tilde{S}_{l}$ defined above.
\begin{enumerate}
\item For an arbitrary index $l \in \overline{1, \, N}$, there exist indices
$k_{1}, \, k_{2} \in \overline{0, \, n-1}$ such that
\begin{equation*}
\tilde{S}_{l}
 = \frac{1}{\sqrt{2}}e_{x,l} + \frac{1}{2}\left(S_{k_{1}} + S_{k_{2}}\right).
\end{equation*}
The linear property of the form $B\left( \cdot , \,  \cdot \right)$ implies
\begin{equation*}
\Phi\left(\tilde{S}_{l}\right)
 = \frac{1}{\sqrt{n}}e_{0,1} + \frac{1}{\sqrt{2}}t_{x,l}e_{x,l}
 + \sum_{i=1}^{n-1}\frac{1}{2}B\left(e_{z,i}, \, S_{k_{1}} + S_{k_{2}}\right)
   t_{z,i}e_{z,i}
\end{equation*}
and
\begin{align*}
2 & = B\left(S_{k_{1}} + S_{k_{2}}, \, S_{k_{1}} + S_{k_{2}}\right) = \\
  & = B\left(\frac{2}{\sqrt{n}}e_{0,1} +
      \sum_{i=1}^{n-1}B\left(e_{z,i}, \, S_{k_{1}} + S_{k_{2}}\right), \,
        S_{k_{1}} + S_{k_{2}}\right) = \\
  & = \frac{4}{n} + \sum_{i=1}^{n-1}\left(B\left(e_{z,i}, \,
        S_{k_{1}} + S_{k_{2}}\right)\right)^{2}.
\end{align*}
Therefore, the Frobenius norm of $\Phi\left(\tilde{S}_{l}\right)$ can be written
in the following form:
\begin{align}
\left|\left| \Phi\left(\tilde{S}_{l}\right) \right|\right|_{F}^{2}
 & = \frac{1}{n} + t_{x,l}^{2} \cdot \left(\frac{1}{\sqrt{2}}\right)^{2} +
     \frac{t_{z}^{2}}{4}\sum_{i=1}^{n-1}\left(B\left(e_{z,i}, \,
       S_{k_{1}} + S_{k_{2}}\right)\right)^{2} = \nonumber \\
 & = \frac{1}{n} + t_{x,l}^{2} \cdot \frac{1}{2} +
     t_{z}^{2} \cdot \frac{1}{4}\left(2-\frac{4}{n}\right).
     \label{eq:diag f norm X}
\end{align}

\item Let $k \in \overline{0, \, n-1}$ be an arbitrary index. The linear
property of the form $B\left( \cdot , \,  \cdot \right)$ implies
\begin{align*}
1 & = B\left(S_{k}, \, S_{k}\right) = \\
  & = B\left(\frac{1}{\sqrt{n}}e_{0,1} + \sum_{i=1}^{n-1}
      B\left(e_{z,i}, \, S_{k}\right)e_{z,i}, \, S_{k}\right) = \\
  & = \frac{1}{n}
    + \sum_{i=1}^{n-1}\left(B\left(e_{z,i}, \, S_{k}\right)\right)^{2}.
\end{align*}
Therefore, the Frobenius norm of $\Phi\left(S_{k}\right)$ can be written
in the following form:
\begin{equation}  \label{eq:diag f norm Z}
\left|\left| \Phi\left(S_{k}\right) \right|\right|_{F}^{2}
 = \frac{1}{n} + t_{z}^{2}\sum_{i=1}^{n}\left(B\left(e_{z,i}, \,
       S_{k}\right)\right)^{2}
 = \frac{1}{n} + t_{z}^{2}\left(1-\frac{1}{n}\right).
\end{equation}

\item Since $\Phi$ is a map with constant Frobenius norm, from
\eqref{eq:diag f norm X} and \eqref{eq:diag f norm Z} it follows that
\begin{equation*}
0 = \left|\left| \Phi\left(\tilde{S}_{l}\right) \right|\right|_{F}^{2}
  - \left|\left| \Phi\left(S_{k}\right) \right|\right|_{F}^{2}
  = \frac{1}{2}t_{x,l}^{2} - \frac{1}{2}t_{z}^{2}.
\end{equation*}
Therefore, $\left| t_{x,l} \right| = \left| t_{z} \right|$ for every
$l \in \overline{1, \, N}$.
\end{enumerate}
\end{enumerate}

The identities above prove that $\left| t_{i} \right| = \left| t_{j} \right|$
for every $i, \, j \in \overline{1, \, n^{2}-1}$.
\end{proof}

\section{\label{sec:gen}Generalizations in dimension $n$}

In this section, we generalize four one-parameter families of $2$-dimensional
channels that appear in the proof of Theorem~\ref{thm:2d classification} and
prove that the generalizations are inequivalent in any dimension $n \ge 3$.

\subsection{\label{ssec:gen def}Definition}

By generalizations of the four $2$-dimensional channels from
Theorem~\ref{thm:2d classification} we mean $n$-dimensional channels that are
channels with constant Frobenius norm and in the case $n = 2$ coincide with the
channels from Theorem~\ref{thm:2d classification}. The idea is to consider the
four $2$-dimensional channels as special cases of $n$-dimensional diagonal
channels, defined in Section~\ref{ssec:diag def}, that multiply the $z$ set of
basis matrices (i.e. $e_{z,1}, \, \dots, \, e_{z,n-1}\in\mathcal{E}$) by $p$ and
multiply the $x$ and $y$ sets of basis matrices (i.e.
$e_{x,1}, \, \dots, \, e_{x,N}\in\mathcal{E}$ and
$e_{y,1}, \, \dots, \, e_{y,N}\in\mathcal{E}$ respectively) by $\pm p$. Then
Theorem~\ref{thm:diag ch cfn crit} insures that the resulting $n$-dimensional
channels are channels with constant Frobenius norm. Therefore, we consider the
following four channels (defined by their form in the basis $\mathcal{E}$):
\begin{itemize}
\item the generalization of the case $\lambda_{x} = \lambda_{y} = p$:
\begin{equation} \label{eq:dep diag form}
\Phi = \mathrm{diag}\left(1, \,
\underbrace{p, \, \dots, \, p}_{N}, \,
\underbrace{p, \, \dots, \, p}_{N}, \,
\underbrace{p, \, \dots, \, p}_{n-1}\right),
\end{equation}
\item the generalization of the case $\lambda_{x} = -\lambda_{y} = p$:
\begin{equation} \label{eq:trd diag form}
\Phi = \mathrm{diag}\left(1, \,
\underbrace{p, \, \dots, \, p}_{N}, \,
\underbrace{-p, \, \dots, \, -p}_{N}, \,
\underbrace{p, \, \dots, \, p}_{n-1}\right),
\end{equation}
\item the generalization of the case $\lambda_{x} = \lambda_{y} = -p$:
\begin{equation} \label{eq:dcq diag form}
\Phi = \mathrm{diag}\left(1, \,
\underbrace{-p, \, \dots, \, -p}_{N}, \,
\underbrace{-p, \, \dots, \, -p}_{N}, \,
\underbrace{p, \, \dots, \, p}_{n-1}\right),
\end{equation}
\item the generalization of the case $\lambda_{x} = -\lambda_{y} = -p$:
\begin{equation} \label{eq:tcq diag form}
\Phi = \mathrm{diag}\left(1, \,
\underbrace{-p, \, \dots, \, -p}_{N}, \,
\underbrace{p, \, \dots, \, p}_{N}, \,
\underbrace{p, \, \dots, \, p}_{n-1}\right).
\end{equation}
\end{itemize}
Plugging the explicit form of the basis matrices into the definitions
\eqref{eq:dep diag form}--\eqref{eq:tcq diag form}, we find that
\begin{itemize}
\item the channel \eqref{eq:dep diag form} is in fact the depolarizing channel:
\begin{equation} \label{eq:dep def}
\Phi_{d}\left(p, \, S\right) = pS + \frac{1-p}{n}\mathrm{Tr} \, S;
\end{equation}
\item the channel \eqref{eq:trd diag form} is in fact the transpose-depolarizing
channel:
\begin{equation} \label{eq:trd def}
\Phi_{t}\left(p, \, S\right) = pS^{T} + \frac{1-p}{n}\mathrm{Tr} \, S;
\end{equation}
\item the channel \eqref{eq:dcq diag form} can be written in the following form:
\begin{equation} \label{eq:dcq def}
\Phi_{dcq}\left(p, \, S\right) = -pS + \frac{1-p}{n}\mathrm{Tr} \, S +
2p\sum_{i=1}^{n} \left\langle \psi_{i} \left| S \right|\psi_{i}\right\rangle
                 \left| \psi_{i} \right\rangle \left\langle \psi_{i} \right|,
\end{equation}
which is why we refer to it as the hybrid
depolarizing-classical-quantum channel (or the DCQ channel for short);
\item the channel \eqref{eq:tcq diag form} can be written in the following form:
\begin{equation} \label{eq:tcq def}
\Phi_{tcq}\left(p, \, S\right) = -pS^{T} + \frac{1-p}{n}\mathrm{Tr} \, S +
2p\sum_{i=1}^{n} \left\langle \psi_{i} \left| S \right| \psi_{i} \right\rangle
                 \left| \psi_{i} \right\rangle \left\langle \psi_{i} \right|,
\end{equation}
which is why we refer to it as the hybrid
transpose-depolarizing-classical-quantum channel (or the TCQ channel for short).
\end{itemize}

Although we refer to the maps \eqref{eq:dep def}--\eqref{eq:tcq def} as
channels, it is the case only for specific values of the parameter $p$, stated
in the following theorem.

\begin{thm} \label{thm:p range}
\begin{itemize}
\item The map $\Phi_{d}$ (defined in \eqref{eq:dep def}) is a quantum channel
if and only if the parameter $p$ satisfies the condition:
\begin{equation} \label{eq:dep p range}
-\frac{1}{n^{2}-1} \le p \le 1.
\end{equation}
\item The map $\Phi_{t}$ (defined in \eqref{eq:trd def}) is a quantum channel
if and only if the parameter $p$ satisfies the condition:
\begin{equation} \label{eq:trd p range}
-\frac{1}{n-1} \le p \le \frac{1}{n+1}.
\end{equation}
\item The map $\Phi_{dcq}$ (defined in \eqref{eq:dcq def}) is a quantum channel
if and only if the parameter $p$ satisfies the condition:
\begin{equation} \label{eq:dcq p range}
-\frac{1}{2n-1} \le p \le \frac{1}{\left(n-1\right)^{2}}.
\end{equation}
\item The map $\Phi_{tcq}$ (defined in \eqref{eq:tcq def}) is a quantum channel
if and only if the parameter $p$ satisfies the condition:
\begin{equation} \label{eq:tcq p range}
-\frac{1}{n-1} \le p \le \frac{1}{n+1}.
\end{equation}
\end{itemize}
\end{thm}

The following lemma, proved in Appendix~\ref{app:pf reprs}, is used to show that
if the conditions \eqref{eq:dep p range}--\eqref{eq:tcq p range} are satisfied,
then the maps \eqref{eq:dep def}--\eqref{eq:tcq def} are quantum channels. As
stated in the proof of Theorem~\ref{thm:p range}, the representations provided
by the lemma are in fact Kraus representations of the four families of channels.

\begin{lem} \label{lem:reprs}
The maps \eqref{eq:dep def}--\eqref{eq:tcq def} can be expressed as
\begin{equation} \label{eq:repr pauli}
\Phi\left(p, \, S\right) =
c_{0} S +
c_{x} \sum_{i=1}^{N} \sigma_{x,i} S \sigma_{x,i} +
c_{y} \sum_{i=1}^{N} \sigma_{y,i} S \sigma_{y,i} +
c_{z} \sum_{i=1}^{N} \sigma_{z,i} S \sigma_{z,i}
\end{equation}
and
\begin{equation} \label{eq:repr basis}
\Phi\left(p, \, S\right) =
\tilde{c}_{0} e_{0,1} S e_{0,1} +
\tilde{c}_{x} \sum_{i=1}^{N} e_{x,i} S e_{x,i} +
\tilde{c}_{y} \sum_{i=1}^{N} e_{y,i} S e_{y,i} +
\tilde{c}_{z} \sum_{j=1}^{n-1} e_{z,j} S e_{z,j},
\end{equation}
where $\sigma_{\alpha,i}$ are the generalized Pauli matrices, $e_{\alpha,i}$ are
the elements of the basis $\mathcal{E}$, and $c_{\alpha}$ and
$\tilde{c}_{\alpha}$ are real coefficients for every
$\alpha \in \left\{ 0, \, x, \, y, \, z \right\}$. More specifically,
\begin{itemize}
\item for the depolarizing channel $\Phi_{d}$,
\begin{equation*}
c_{0} = \frac{1+\left(n^{2}-1\right)p}{n^{2}}, \quad
c_{x} = c_{y} = \frac{1-p}{2n}, \quad
c_{z} = \frac{1-p}{n^{2}},
\end{equation*}
\begin{equation*}
\tilde{c}_{0} = \frac{1+\left(n^{2}-1\right)p}{n}, \quad
\tilde{c}_{x} = \tilde{c}_{y} = \tilde{c}_{z} = \frac{1-p}{n};
\end{equation*}
\item for the transpose-depolarizing channel $\Phi_{t}$,
\begin{equation*}
c_{0} = c_{z} = \frac{1+\left(n-1\right)p}{n^{2}}, \quad
c_{x} = \frac{1+\left(n-1\right)p}{2n}, \quad
c_{y} = \frac{1-\left(n+1\right)p}{2n},
\end{equation*}
\begin{equation*}
\tilde{c}_{0} = \tilde{c}_{x} = \tilde{c}_{z} =
  \frac{1+\left(n-1\right)p}{n}, \quad
\tilde{c}_{y} = \frac{1-\left(n+1\right)p}{n};
\end{equation*}
\item for the DCQ channel $\Phi_{dcq}$,
\begin{equation*}
c_{0} = \frac{1-\left(n-1\right)^{2}p}{n^{2}}, \quad
c_{x} = c_{y} = \frac{1-p}{2n}, \quad
c_{z} = \frac{1+\left(2n-1\right)p}{n^{2}},
\end{equation*}
\begin{equation*}
\tilde{c}_{0} = \frac{1-\left(n-1\right)^{2}p}{n}, \quad
\tilde{c}_{x} = \tilde{c}_{y} = \frac{1-p}{n}, \quad
\tilde{c}_{z} = \frac{1+\left(2n-1\right)p}{n};
\end{equation*}
\item for the TCQ channel $\Phi_{tcq}$,
\begin{equation*}
c_{0} = c_{z} = \frac{1+\left(n-1\right)p}{n^{2}}, \quad
c_{x} = \frac{1-\left(n+1\right)p}{2n}, \quad
c_{y} = \frac{1+\left(n-1\right)p}{2n},
\end{equation*}
\begin{equation*}
\tilde{c}_{0} = \tilde{c}_{y} = \tilde{c}_{z} =
  \frac{1+\left(n-1\right)p}{n}, \quad
\tilde{c}_{x} = \frac{1-\left(n+1\right)p}{n}.
\end{equation*}
\end{itemize}
\end{lem}

\begin{proof}[Proof of Theorem~\ref{thm:p range}]
Since \eqref{eq:dep p range} and \eqref{eq:trd p range} are well-known results
(e.g. present in \cite{King2003} and \cite{Datta2006}), we omit the respective
proofs, although one can use the method described below to prove
\eqref{eq:dep p range} and \eqref{eq:trd p range} similarly to
\eqref{eq:dcq p range} and \eqref{eq:tcq p range}.

First, we prove sufficiency of the conditions \eqref{eq:dcq p range} and
\eqref{eq:tcq p range} using Lemma~\ref{lem:reprs} and Kraus' theorem. Suppose
that the condition \eqref{eq:dcq p range} (respectively, \eqref{eq:tcq p range})
is satisfied. Then all coefficients in representations \eqref{eq:repr pauli} and
\eqref{eq:repr basis} are nonnegative. Therefore, \eqref{eq:repr pauli} and
\eqref{eq:repr basis} are Kraus representations, hence $\Phi_{dcq}$
(respectively, $\Phi_{tcq}$) is a quantum channel.

Second, we prove necessity of the conditions \eqref{eq:dcq p range} and
\eqref{eq:tcq p range} using Choi's theorem.

\begin{itemize}

\item Suppose that the map $\Phi_{dcq}$ is a quantum channel. Then
$\Phi_{dcq}^{*}$ is a completely positive map and the Choi matrix
\begin{equation*}\scalebox{.75}{$
C = \left(\begin{array}{cccc|cccc|c|cccc}
p + \frac{1-p}{n} & 0 & \dots & 0 & 0 & -p & \dots & 0 & \dots & 0 & 0 & \dots & -p \\
0 & \frac{1-p}{n} & \dots & 0 & 0 & 0 & \dots & 0 & \dots & 0 & 0 & \dots & 0 \\
\dots & \dots & \dots & \dots & \dots & \dots & \dots & \dots & \dots & \dots & \dots & \dots & \dots \\
0 & 0 & \dots & \frac{1-p}{n} & 0 & 0 & \dots & 0 & \dots & 0 & 0 & \dots & 0 \\
\hline
0 & 0 & \dots & 0 & \frac{1-p}{n} & 0 & \dots & 0 & \dots & 0 & 0 & \dots & 0 \\
-p & 0 & \dots & 0 & 0 & p + \frac{1-p}{n} & \dots & 0 & \dots & 0 & 0 & \dots & -p \\
\dots & \dots & \dots & \dots & \dots & \dots & \dots & \dots & \dots & \dots & \dots & \dots & \dots \\
0 & 0 & \dots & 0 & 0 & 0 & \dots & \frac{1-p}{n} & \dots & 0 & 0 & \dots & 0 \\
\hline
\dots & \dots & \dots & \dots & \dots & \dots & \dots & \dots & \dots & \dots & \dots & \dots & \dots \\
\hline
0 & 0 & \dots & 0 & 0 & 0 & \dots & 0 & \dots & \frac{1-p}{n} & \dots & 0 & 0 \\
0 & 0 & \dots & 0 & 0 & 0 & \dots & 0 & \dots & \dots & \dots & \dots & \dots \\
\dots & \dots & \dots & \dots & \dots & \dots & \dots & \dots & \dots & 0 & \dots & \frac{1-p}{n} & 0 \\
-p & 0 & \dots & 0 & 0 & -p & \dots & 0 & \dots & 0 & \dots & 0 & p + \frac{1-p}{n}
\end{array}\right)
$}\end{equation*}
is positive semidefinite by Choi's theorem. In particular, the element $C_{2,2}$
is nonnegative: $\frac{1-p}{n} \ge 0$, hence the parameter $p$ satisfies
$p \le 1$. It remains to show that $p \ge -\frac{1}{n^{2}-1}$. Since $C$ is
positive semidefinite, we have $\det C \ge 0$. Expanding the determinant using
the elements of the main diagonal, we get:
\begin{equation*}
\det C = \left(\frac{1-p}{n}\right)^{n\left(n-1\right)} \cdot
\det \left(\begin{array}{cccc}
p + \frac{1-p}{n} & -p                & \dots & -p \\
-p                & p + \frac{1-p}{n} & \dots & -p \\
\dots             & \dots             & \dots & \dots \\
-p                & -p                & \dots & p + \frac{1-p}{n}
\end{array}\right).
\end{equation*}
To compute the determinant on the right hand side, consider
\begin{equation*}\scalebox{.9}{$
D\left(a, \, b\right) = \det \left(\underbrace{\begin{array}{cccc}
p + \frac{1-p}{n} & -p                & \dots & -p \\
-p                & p + \frac{1-p}{n} & \dots & -p \\
\dots             & \dots             & \dots & \dots \\
-p                & -p                & \dots & p + \frac{1-p}{n} \\
-p                & -p                & \dots & -p \\
-p                & -p                & \dots & -p \\
\dots             & \dots             & \dots & \dots \\
-p                & -p                & \dots & -p
\end{array}}_{a}\underbrace{\begin{array}{cccc}
-p    & -p    & \dots & -p \\
-p    & -p    & \dots & -p \\
\dots & \dots & \dots & \dots \\
-p    & -p    & \dots & -p \\
-p    & -p    & \dots & -p \\
-p    & -p    & \dots & -p \\
\dots & \dots & \dots & \dots \\
-p    & -p    & \dots & -p
\end{array}}_{b}\right).
$}\end{equation*}
Using the multilinearity property of the determinant, we obtain the following
recurrence relation:
\begin{equation*}
D\left(a, \, b\right) = D\left(a - 1, \, b + 1\right) +
    \left(2p + \frac{1-p}{n}\right) D\left(a - 1, \, b\right).
\end{equation*}
Solving this recurrence relation, we find
\begin{equation*}
D\left(n, \, 0\right) = \left(2p + \frac{1-p}{n}\right)^{n-1} \cdot
    \frac{1 - \left(n-1\right)^{2}p}{n}.
\end{equation*}
Since
\begin{equation*}
\det C =
\left(\frac{1-p}{n}\right)^{n\left(n-1\right)} \cdot D\left(n, \, 0\right),
\end{equation*}
the inequality $\det C \ge 0$ holds if and only if
\begin{equation} \label{eq:dcq cpos ineq 1}
\left(\frac{1+\left(2n-1\right)p}{n}\right)^{n-1} \cdot
\frac{1-\left(n-1\right)^{2}p}{n} \ge 0.
\end{equation}
A similar calculation for the submatrix
$\left(C_{i,j}\right)_{i,j=1}^{n\left(n - 1\right)}$ yields
\begin{equation} \label{eq:dcq cpos ineq 2}
\left(\frac{1+\left(2n-1\right)p}{n}\right)^{n-2} \cdot
\frac{1-\left(n-1\right)^{2}p}{n} \ge 0.
\end{equation}
The inequalities \eqref{eq:dcq cpos ineq 1} and \eqref{eq:dcq cpos ineq 2} hold
simultaneously if and only if the following conditions are satisfied:
\begin{equation*}
\frac{1+\left(2n-1\right)p}{n} \ge 0, \quad
\frac{1-\left(n-1\right)^{2}p}{n} \ge 0,
\end{equation*}
which is equivalent to the condition \eqref{eq:dep p range}.

\item Suppose that the map $\Phi_{tcq}$ is a quantum
channel. Then $\Phi_{tcq}^{*}$ is a completely positive map. By Choi's theorem,
the Choi matrix
\begin{equation*}\scalebox{.75}{$
C = \left(\begin{array}{cccc|cccc|c|cccc}
p + \frac{1-p}{n} & 0 & \dots & 0 & 0 & 0 & \dots & 0 & \dots & 0 & 0 & \dots & 0 \\
0 & \frac{1-p}{n} & \dots & 0 & -p & 0 & \dots & 0 & \dots & 0 & 0 & \dots & 0 \\
\dots & \dots & \dots & \dots & \dots & \dots & \dots & \dots & \dots & \dots & \dots & \dots & \dots \\
0 & 0 & \dots & \frac{1-p}{n} & 0 & 0 & \dots & 0 & \dots & -p & 0 & \dots & 0 \\
\hline
0 & -p & \dots & 0 & \frac{1-p}{n} & 0 & \dots & 0 & \dots & 0 & 0 & \dots & 0 \\
0 & 0 & \dots & 0 & 0 & p + \frac{1-p}{n} & \dots & 0 & \dots & 0 & 0 & \dots & 0 \\
\dots & \dots & \dots & \dots & \dots & \dots & \dots & \dots & \dots & \dots & \dots & \dots & \dots \\
0 & 0 & \dots & 0 & 0 & 0 & \dots & \frac{1-p}{n} & \dots & 0 & -p & \dots & 0 \\
\hline
\dots & \dots & \dots & \dots & \dots & \dots & \dots & \dots & \dots & \dots & \dots & \dots & \dots \\
\hline
0 & 0 & \dots & -p & 0 & 0 & \dots & 0 & \dots & \frac{1-p}{n} & \dots & 0 & 0 \\
0 & 0 & \dots & 0 & 0 & 0 & \dots & -p & \dots & \dots & \dots & \dots & \dots \\
\dots & \dots & \dots & \dots & \dots & \dots & \dots & \dots & \dots & 0 & \dots & \frac{1-p}{n} & 0 \\
0 & 0 & \dots & 0 & 0 & 0 & \dots & 0 & \dots & 0 & \dots & 0 & p + \frac{1-p}{n}
\end{array}\right)
$}\end{equation*}
is positive semidefinite. In particular,
\begin{equation*}
\det \left(\begin{array}{cc} C_{2,2} & C_{2,n+1} \\
                             C_{n+1,2} & C_{n+1,n+1} \end{array}\right) =
\left(\frac{1}{n^{2}}-1\right)p^{2} - \frac{2}{n^{2}}p + \frac{1}{n^{2}} \ge 0.
\end{equation*}
Therefore,
\begin{equation*}
\left(1 + \left(n - 1\right)p\right) \left(1 - \left(n + 1\right)p\right) \ge 0,
\end{equation*}
which holds if and only if the parameter $p$ satisfies the condition
\eqref{eq:tcq p range}.
\end{itemize}
\end{proof}

\subsection{\label{ssec:gen ineq}Inequivalence in dimension $n \ge 3$}

In this section, we prove the following theorem.

\begin{thm} \label{thm:ineqiv}
In dimension $n \ge 3$, the one-parameter families of channels $\Phi_{d}$,
$\Phi_{t}$, $\Phi_{dcq}$, and $\Phi_{tcq}$ are not equivalent (unlike in
dimension $n = 2$, as shown in Section~\ref{sec:dim 2}).
\end{thm}

To prove this theorem, we use the following two lemmas.

\begin{lem} \label{lem:ch eigenvals}
Let $S_{1}$ and $S_{2}$ be two states with the same set of eigenvalues. Then
$\Phi_{d}\left(p,\,S_{1}\right)$ and $\Phi_{d}\left(p,\,S_{2}\right)$,
$\Phi_{t}\left(p,\,S_{1}\right)$ and $\Phi_{t}\left(p,\,S_{2}\right)$ have the
same set of eigenvalues;
$\Phi_{dcq}\left(p,\,S_{1}\right)$ and $\Phi_{dcq}\left(p,\,S_{2}\right)$,
$\Phi_{tcq}\left(p,\,S_{1}\right)$ and $\Phi_{tcq}\left(p,\,S_{2}\right)$ may
have different sets of eigenvalues, unless $p = 0$ or $n = 2$.
\end{lem}

\begin{proof}
Denote $\left\{\lambda_{k}\right\}_{k=1}^{n}$ the eigenvalues of the states
$S_{1}$ and $S_{2}$. From the definitions of the depolarizing channel and the
transpose-depolarizing channel, we can see that the output states
$\Phi_{d}\left(p,\,S_{1}\right)$, $\Phi_{d}\left(p,\,S_{2}\right)$,
$\Phi_{t}\left(p,\,S_{1}\right)$ and $\Phi_{t}\left(p,\,S_{2}\right)$ all have
the same eigenvalues $\left\{ p\lambda_{k} + \frac{1-p}{n}\right\} _{k=1}^{n}$.

For the DCQ and the TCQ channels, consider two states $S_{1}$ and $S_{2}$
defined as follows:
\begin{equation*}
\left| \eta \right\rangle = \left(\begin{array}{c}
        1 \\ 0 \\ \dots \\ 0 \end{array}\right), \quad
S_{1} = \left| \eta \right\rangle \left\langle \eta \right|
      = \left(\begin{array}{cccc}
        1     & 0     & \dots & 0 \\
        0     & 0     & \dots & 0 \\
        \dots & \dots & \dots & \dots \\
        0     & 0     & \dots & 0
        \end{array}\right),
\end{equation*}
\begin{equation*}
\left| \zeta \right\rangle = \frac{1}{\sqrt{n}}\left(\begin{array}{c}
        1 \\ 1 \\ \dots \\ 1 \end{array}\right), \quad
S_{2} = \left| \zeta \right\rangle \left\langle \zeta \right|
      = \frac{1}{n}\left(\begin{array}{cccc}
        1     & 1     & \dots & 1 \\
        1     & 1     & \dots & 1 \\
        \dots & \dots & \dots & \dots \\
        1     & 1     & \dots & 1
        \end{array}\right).
\end{equation*}
Since $S_{1}$ and $S_{2}$ are pure states, they have eigenvalues
$\lambda_{1} = 1$ and $\lambda_{2} = \dots = \lambda_{n} = 0$. However,
$\Phi_{dcq}\left(p,\,S_{1}\right)$ and $\Phi_{tcq}\left(p,\,S_{1}\right)$ have
eigenvalues $\lambda_{1} = p + \frac{1-p}{n}$ and
$\lambda_{2} = \dots = \lambda_{n} = \frac{1-p}{n}$, while
$\Phi_{dcq}\left(p,\,S_{2}\right)$ and $\Phi_{tcq}\left(p,\,S_{2}\right)$ have
eigenvalues $\lambda_{1} = -p + \frac{1-p}{n}$ and
$\lambda_{2} = \dots = \lambda_{n} = -p + \frac{1+p}{n}$.
These sets of eigenvalues are the same only in two cases: if $p = 0$, which is
trivial, or if $n = 2$, which we already considered in Section~\ref{sec:dim 2}.
\end{proof}

\begin{lem} \label{lem:equiv mul alpha}
Suppose that for some $p$ and $\tilde{p}$ the channels
$\Phi_{d}\left(p,\,S\right)$ and $\Phi_{t}\left(\tilde{p},\,S\right)$ are
equivalent. Let $\alpha\in\mathbb{R}$ be a multiplier such that $\alpha p$ and
$\alpha\tilde{p}$ are in the valid parameter range for $\Phi_{d}$ and $\Phi_{t}$
respectively. Then the channels $\Phi_{d}\left(\alpha p,\,S\right)$ and
$\Phi_{t}\left(\alpha\tilde{p},\,S\right)$ are equivalent.

The same holds for the DCQ and the TCQ channels: if for some $p$ and $\tilde{p}$
we have $\Phi_{dcq}\left(p,\,S\right) \sim
\Phi_{tcq}\left(\tilde{p},\,S\right)$, then $\Phi_{dcq}\left(\alpha p,\,S\right)
\sim \Phi_{tcq}\left(\alpha\tilde{p},\,S\right)$ for every $\alpha\in\mathbb{R}$
such that $\alpha p$ and $\alpha\tilde{p}$ are in the valid parameter range for
$\Phi_{dcq}$ and $\Phi_{tcq}$ respectively.
\end{lem}

\begin{proof}
Suppose that for some $p$ and $\tilde{p}$ we have
$\Phi_{d}\left(p,\,S\right) \sim \Phi_{t}\left(\tilde{p},\,S\right)$. Then
there exist unitary operators $U_{1}$ and $U_{2}$ such that for every quantum
state $S$ the following identity holds:
\begin{equation*}
\Phi_{d}\left(p, \, S\right)
 = U_{2}\Phi_{t}\left(\tilde{p}, \, U_{1} S U_{1}^{*}\right)U_{2}^{*}.
\end{equation*}
Let $\alpha\in\mathbb{R}$ be a multiplier such that $\alpha p$ and
$\alpha\tilde{p}$ are in the valid parameter range for $\Phi_{d}$ and $\Phi_{t}$
respectively. Note that
\begin{equation*}
\Phi_{d}\left(\alpha p, \, S\right)
 = \alpha\Phi_{d}\left(p, \, S\right)
 + \frac{1-\alpha}{n}\mathrm{Tr} \, S,
\end{equation*}
\begin{equation*}
\Phi_{t}\left(\alpha p, \, S\right)
 = \alpha\Phi_{t}\left(p, \, S\right)
 + \frac{1-\alpha}{n}\mathrm{Tr} \, S.
\end{equation*}
Therefore,
\begin{align*}
\Phi_{d}\left(\alpha p, \, S\right)
 & = \alpha\Phi_{d}\left(p, \, S\right)
   + \frac{1-\alpha}{n}\mathrm{Tr} \, S
   = \alpha U_{2}\Phi_{t}\left(\tilde{p}, \, U_{1} S U_{1}^{*}\right)U_{2}^{*}
   + \frac{1-\alpha}{n}\mathrm{Tr} \, S = \\
 & = U_{2}\left\{ \alpha\Phi_{t}\left(\tilde{p}, \, U_{1} S U_{1}^{*}\right)
   + \frac{1-\alpha}{n}\mathrm{Tr} \, S\right\} U_{2}^{*}
   = U_{2}\Phi_{t}\left(\alpha\tilde{p}, \, U_{1} S U_{1}^{*}\right)U_{2}^{*}.
\end{align*}
Thus, equivalence $\Phi_{d}\left(\alpha p,\,S\right) \sim
\Phi_{t}\left(\alpha\tilde{p},\,S\right)$ holds. We get the proof for the DCQ
and the TCQ channels by replacing $\Phi_{d}$ with $\Phi_{dcq}$ and $\Phi_{t}$
with $\Phi_{tcq}$.
\end{proof}

\begin{proof}[Proof of Theorem~\ref{thm:ineqiv}]
Since unitary evolutions do not change the set of eigenvalues,
the following inequivalences hold by Lemma~\ref{lem:ch eigenvals}
(for $n \ge 3$):
$\Phi_{d} \not \sim \Phi_{dcq}$, $\Phi_{d} \not \sim \Phi_{tcq}$,
$\Phi_{t} \not \sim \Phi_{dcq}$, $\Phi_{t} \not \sim \Phi_{tcq}$. Thus it
remains to prove that $\Phi_{d} \not \sim \Phi_{t}$ and
$\Phi_{dcq} \not \sim \Phi_{tcq}$ in any dimension $n \ge 3$.

Suppose that $\Phi_{d} \sim \Phi_{t}$. Note that the case when $p = 0$ and
$\tilde{p} \neq 0$ is impossible, since a completely depolarizing channel cannot
be equivalent to a channel that is not a constant mapping. Similarly, the case
when $p \neq 0$ and $\tilde{p} = 0$ is also impossible. Therefore, without loss
of generality we may assume that for some $p \neq 0$ and $\tilde{p} \neq 0$ we
have $\Phi_{d}\left(p,\,S\right) \sim \Phi_{t}\left(\tilde{p},\,S\right)$. Let
$\alpha\in\mathbb{R}$ be a multiplier such that $\alpha p$ and $\alpha\tilde{p}$
are in the valid parameter range for $\Phi_{d}$ and $\Phi_{t}$ respectively. By
Lemma~\ref{lem:equiv mul alpha}, the channels
$\Phi_{d}\left(\alpha p,\,S\right)$ and
$\Phi_{t}\left(\alpha\tilde{p},\,S\right)$ are equivalent for suitable values of
$\alpha$. On the one hand, $\Phi_{d}\left(\alpha p,\,S\right)$ is required to be
a channel, so $\alpha$ lies between $-\frac{1}{\left(n^{2}-1\right)p}$ and
$\frac{1}{p}$. On the other hand, $\Phi_{t}\left(\alpha p,\,S\right)$ is
required to be a channel, so $\alpha$ lies between
$-\frac{1}{\left(n-1\right)\tilde{p}}$ and
$\frac{1}{\left(n+1\right)\tilde{p}}$. Since, a unitary evolution of a quantum
channel is a quantum channel, the upper and the lower bound for $\alpha$ we
obtained in two ways are the same. Now, consider the following two cases.
\begin{enumerate}
\item The parameters $p$ and $\tilde{p}$ are of the same sign. The equality of
the two upper and the two lower bound for $\alpha$ yeilds the following system:
\begin{equation*}
\begin{cases}
-\frac{1}{\left(n^{2}-1\right)p} = -\frac{1}{\left(n-1\right)\tilde{p}}, \\
\frac{1}{p} = \frac{1}{\left(n+1\right)\tilde{p}}.
\end{cases}
\end{equation*}
These equalities hold only if either $n = 0$ or $n = -2$, which is impossible.
\item The parameters $p$ and $\tilde{p}$ are of different signs. Then
\begin{equation*}
\begin{cases}
-\frac{1}{\left(n^{2}-1\right)p} = \frac{1}{\left(n+1\right)\tilde{p}}, \\
\frac{1}{p} = -\frac{1}{\left(n-1\right)\tilde{p}}.
\end{cases}
\end{equation*}
These equalities hold only if either $n = 0$, which is impossible, or $n = 2$,
which we already considered in Section~\ref{sec:dim 2}.
\end{enumerate}
Thus, $\Phi_{d} \not \sim \Phi_{t}$ in any dimension $n \ge 3$.

The same reasoning can be used for the DCQ and the TCQ channels.
\begin{enumerate}
\item If the parameters $p$ and $\tilde{p}$ are of the same sign, then the
equality of the two upper and the two lower bound for $\alpha$ yeilds the
following system:
\begin{equation*}
\begin{cases}
-\frac{1}{\left(2n-1\right)p} = -\frac{1}{\left(n-1\right)\tilde{p}}, \\
\frac{1}{\left(n-1\right)^{2}p} = \frac{1}{\left(n+1\right)\tilde{p}}.
\end{cases}
\end{equation*}
These equalities hold only if either $n = 0$ or $n = \frac{5\pm\sqrt{17}}{2}$,
which is impossible.
\item If the parameters $p$ and $\tilde{p}$ are of different signs, then
\begin{equation*}
\begin{cases}
-\frac{1}{\left(2n-1\right)p} = \frac{1}{\left(n+1\right)\tilde{p}}, \\
\frac{1}{\left(n-1\right)^{2}p} = -\frac{1}{\left(n-1\right)\tilde{p}}.
\end{cases}
\end{equation*}
These equalities hold only if either $n = 0$, which is impossible, or $n = 2$,
which we already considered in Section~\ref{sec:dim 2}.
\end{enumerate}
Thus, $\Phi_{dcq} \not \sim \Phi_{tcq}$ in any dimension $n \ge 3$.
\end{proof}

\section{\label{sec:conclusion}Concluding remarks}

The main result of this paper are the examples of channels with constant
Frobenius norm --- namely, $\Phi_{dcq}$ and $\Phi_{tcq}$ --- that are equivalent
neither to the depolarizing channel $\Phi_{d}$ nor to the transpose-depolarizing
channel $\Phi_{t}$ in dimension $n \ge 3$, although all four channels are
equivalent in dimension $n = 2$. It is also worth noting that the intuitive
generalization of Pauli matrices and the basis $\mathcal{E}$ in
$\mathcal{M}_{n}$, which were used to define diagonal channels, are different
from the generalized Pauli matrices and the basis used in \cite{Nathanson2007}.

Although some progress has been made in this work, there are many open questions
about diagonal channels and sets of channels with constant functions.

\begin{itemize}
\item First, it remains to describe all diagonal channels with constant
Frobenius norm in dimensions $n \ge 3$, since we only fully covered the case
$n = 2$. Intermediate results in dimension $3$ show that more families of
inequivalent diagonal channels emerge and that the parameter ranges become more
complex (for example, the parameter range for one famility of $3$-dimensional
diagonal channels is $\frac{1}{26}\left(-1-3\sqrt{3}\right) \le p \le
\frac{1}{26}\left(-1+3\sqrt{3}\right)$, while the bound of the parameter range
for another family are the two largest roots of the polynomial
$-1+21x^{2}+7x^{3}$). Since complexity of calculations grows as dimension $n$
increases, using more advanced methods and theories, such as algebraic geometry,
might become necessary in higher dimensions and in the general case.
\item Second, it is unknown if the set of channels with constant Frobenius norm
in an arbitrary dimension is limited to completely depolarizing and diagonal
channels (as in the $2$-dimensional case).
\item Finally, almost nothing is known about sets of channels with constant von
Neumann entropy and other matrix norms (outside the $2$-dimensional case) and
their relationship with the set of channels with constant Frobenius norm.
\end{itemize}

In the future, we hope to tackle these and other problems and generalize the
results obtained in this paper.

\section*{Acknowledgement}

I would like to thank my adviser professor Grigori G. Amosov for formulating the
problem, fruitful discussions, and helpful advice.

\begin{appendices}
\section{\label{app:pf reprs}Proof of Lemma~\ref{lem:reprs}}

This section contains the proof of Lemma~\ref{lem:reprs} (from
Section~\ref{ssec:gen def}), which gives two representations for each of the
four generalizations of diagonal channels. The proof relies on the following
technical lemma.

\begin{lem} \label{lem:tech}
Let $S = \left(s_{ij}\right)_{i,j=1}^{n} \in \mathcal{M}_{n}$ be an arbitrary
matrix. Then the following identities hold:
\begin{align}
\sum_{i=1}^{N} \sigma_{x,i} S \sigma_{x,i} & =
  S + \mathrm{Tr}\,S - 2\cdot\mathrm{diag}\left(s_{11},\,\dots,\,s_{nn}\right),
  \label{eq:tec sum pauli x} \\
\sum_{i=1}^{N} \sigma_{y,i} S \sigma_{y,i} & =
  \mathrm{Tr}\,S - S,
  \label{eq:tec sum pauli y} \\
\sum_{i=1}^{N} \sigma_{z,i} S \sigma_{z,i} & =
  n\cdot\mathrm{diag}\left(s_{11},\,\dots,\,s_{nn}\right) - S^{T},
  \label{eq:tec sum pauli z} \\
\sum_{i=1}^{N} e_{x,i} S e_{x,i} & =
  \frac{1}{2}\left(S + \mathrm{Tr}\,S
  - 2\cdot\mathrm{diag}\left(s_{11},\,\dots,\,s_{nn}\right)\right),
  \label{eq:tec sum basis x} \\
\sum_{i=1}^{N} e_{y,i} S e_{y,i} & =
  \frac{1}{2}\left(\mathrm{Tr}\,S - S\right),
  \label{eq:tec sum basis y} \\
\sum_{i=1}^{n-1} e_{z,i} S e_{z,i} & =
  \frac{1}{n}\left(n\cdot\mathrm{diag}\left(s_{11},\,\dots,\,s_{nn}\right)
  - S^{T}\right).
  \label{eq:tec sum basis z}
\end{align}
\end{lem}

\begin{proof}
First, we calculate the following products:
\begin{equation*}\scalebox{.7}{$
\sigma_{x,1} S \sigma_{x,1} = \left(\begin{array}{ccccc}
s_{22} & s_{21} & 0     & \dots & 0 \\
s_{12} & s_{11} & 0     & \dots & 0 \\
0      & 0      & 0     & \dots & 0 \\
\dots  & \dots  & \dots & \dots & \dots \\
0      & 0      & 0     & \dots & 0
\end{array}\right), \, \dots, \,
\sigma_{x,N} S \sigma_{x,N} = \left(\begin{array}{ccccc}
0     & \dots & 0     & 0         & 0 \\
\dots & \dots & \dots & \dots     & \dots \\
0     & \dots & 0     & 0         & 0 \\
0     & \dots & 0     & s_{n,n}   & s_{n,n-1} \\
0     & \dots & 0     & s_{n-1,n} & s_{n-1,n-1}
\end{array}\right);
$}\end{equation*}
\begin{equation*}\scalebox{.7}{$
\sigma_{y,1} S \sigma_{y,1} = \left(\begin{array}{ccccc}
s_{22}  & -s_{21} & 0     & \dots & 0 \\
-s_{12} & s_{11}  & 0     & \dots & 0 \\
0       & 0       & 0     & \dots & 0 \\
\dots   & \dots   & \dots & \dots & \dots \\
0       & 0       & 0     & \dots & 0
\end{array}\right), \, \dots, \,
\sigma_{y,N} S \sigma_{y,N} = \left(\begin{array}{ccccc}
0     & \dots & 0     & 0          & 0 \\
\dots & \dots & \dots & \dots      & \dots \\
0     & \dots & 0     & 0          & 0 \\
0     & \dots & 0     & s_{n,n}    & -s_{n,n-1} \\
0     & \dots & 0     & -s_{n-1,n} & s_{n-1,n-1}
\end{array}\right),
$}\end{equation*}
\begin{equation*}\scalebox{.7}{$
\sigma_{z,1} S \sigma_{z,1} = \left(\begin{array}{ccccc}
s_{11}  & -s_{12} & 0     & \dots & 0 \\
-s_{21} & s_{22}  & 0     & \dots & 0 \\
0       & 0       & 0     & \dots & 0 \\
\dots   & \dots   & \dots & \dots & \dots \\
0       & 0       & 0     & \dots & 0
\end{array}\right), \, \dots, \,
\sigma_{z,N} S \sigma_{z,N} = \left(\begin{array}{ccccc}
0     & \dots & 0     & 0           & 0 \\
\dots & \dots & \dots & \dots       & \dots \\
0     & \dots & 0     & 0           & 0 \\
0     & \dots & 0     & s_{n-1,n-1} & -s_{n-1,n} \\
0     & \dots & 0     & -s_{n,n-1}  & s_{n,n}
\end{array}\right).
$}\end{equation*}
Adding up these products, we obtain the required sums
\eqref{eq:tec sum pauli x}--\eqref{eq:tec sum pauli z}.

Since $e_{x,i} = \frac{1}{\sqrt{2}}\sigma_{x,i}$ and
$e_{y,i} = \frac{1}{\sqrt{2}}\sigma_{y,i}$, we can use the results for
$\sigma_{x,i}$ and $\sigma_{y,i}$ to obtain the sums \eqref{eq:tec sum basis x}
and \eqref{eq:tec sum basis y}. We prove \eqref{eq:tec sum basis z} by
calculating the following products, similar to those computed above:
\begin{equation*}\scalebox{.72}{$
e_{z,1} S e_{z,1} = \frac{1}{2}\left(\begin{array}{ccccc}
s_{11}  & -s_{12} & 0     & \dots & 0 \\
-s_{21} & s_{22}  & 0     & \dots & 0 \\
0       & 0       & 0     & \dots & 0 \\
\dots   & \dots   & \dots & \dots & \dots \\
0       & 0       & 0     & \dots & 0
\end{array}\right), \,
e_{z,2} S e_{z,2} = \frac{1}{6}\left(\begin{array}{ccccc}
s_{11}   & s_{12}   & -2s_{13} & \dots & 0 \\
s_{21}   & s_{22}   & -2s_{23} & \dots & 0 \\
-2s_{31} & -2s_{32} & 4s_{33}  & \dots & 0 \\
\dots    & \dots    & \dots    & \dots & \dots \\
0        & 0        & 0        & \dots & 0
\end{array}\right), \, \dots,
$}\end{equation*}
\begin{equation*}\scalebox{.72}{$
e_{z,n-1} S e_{z,n-1} = \frac{1}{\left(n-1\right)n}\left(\begin{array}{ccccc}
   s_{1,1}   & \dots & s_{1,n-2}   & s_{1,n-1}   & -\left(n-1\right)s_{1,n} \\
   \dots     & \dots & \dots       & \dots       & \dots \\
   s_{n-2,1} & \dots & s_{n-2,n-2} & s_{n-2,n-2} & -\left(n-1\right)s_{n-2,n} \\
   s_{n-1,1} & \dots & s_{n-1,n-2} & s_{n-1,n-1} & -\left(n-1\right)s_{n-1,n} \\
   -\left(n-1\right)s_{n,1} & \dots & -\left(n-1\right)s_{n,n-2} &
     -\left(n-1\right)s_{n,n-1} & \left(n-1\right)^{2}s_{n,n}
\end{array}\right).
$}\end{equation*}
Adding up the products and using the identity
\begin{equation*}
\sum_{i=1}^{l}\frac{1}{i\left(i+1\right)} = \frac{l}{l+1},
\end{equation*}
which can be proved by induction, we find:
\begin{equation*}
\left(\sum_{i=1}^{n-1} e_{z,i} S e_{z,i}\right)_{k,k}
 = \left[\frac{\left(k-1\right)^{2}}{\left(k-1\right)k}
 + \sum_{i=k}^{n-1}\frac{1}{i\left(i+1\right)}\right]s_{k,k}
 = \frac{n-1}{n}s_{k,k},
\end{equation*}
\begin{equation*}
\left(\sum_{i=1}^{n-1} e_{z,i} S e_{z,i}\right)_{k,l}
 = \left[\frac{-\left(l-1\right)}{\left(l-1\right)l}
 + \sum_{i=l}^{n-1}\frac{1}{i\left(i+1\right)}\right]s_{k,l}
 = -\frac{1}{n}s_{k,l}.
\end{equation*}
\end{proof}

\begin{proof}[Proof of Lemma~\ref{lem:reprs}]
To prove the representation of the depolarizing channel with Pauli matrices,
consider the channel of the form \eqref{eq:repr pauli} with coefficients
$c_{\alpha,i}$ as in the statement of the theorem. Applying
Lemma~\ref{lem:tech}, we find:
\begin{equation*}\scalebox{.9}{$
\left(\Phi\left(S\right)\right)_{ii}
 = \frac{1+\left(n^{2}-1\right)p}{n^{2}}s_{ii}
 + 2 \cdot \frac{1-p}{2n} \left(\mathrm{Tr} \, S - s_{ii}\right)
 + \frac{1-p}{n^{2}} \left(n-1\right)s_{ii}
 = ps_{ii} + \frac{1-p}{n}\mathrm{Tr} \, S,
$}\end{equation*}
\begin{equation*}
\left(\Phi\left(S\right)\right)_{ij}
 = \frac{1+\left(n^{2}-1\right)p}{n^{2}}s_{ij}
 + \frac{1-p}{2n}\left(s_{ji} - s_{ji}\right)
 - \frac{1-p}{n^{2}}s_{ij}
 = ps_{ij}.
\end{equation*}
Therefore,
\begin{equation*}
\Phi\left(p, \, S\right)
 = pS + \frac{1-p}{n}\mathrm{Tr} \, S
 = \Phi_{d}\left(p, \, S\right),
\end{equation*}
which proves the representation. For all other representations the proof follows
the same plan.
\end{proof}

\end{appendices}

\end{document}